\documentclass[a4paper,USenglish,autoref,cleveref]{lipics-v2021}

\usepackage[ruled,vlined,linesnumberedhidden]{algorithm2e}
\usepackage{graphicx}\graphicspath{{./figures}}

\usepackage{amsthm}
\usepackage{amsmath}
\usepackage{amssymb}
\usepackage{soul}
\usepackage{tikz}

\title{The Parametrized Complexity of~the~Segment~Number}

\author{Sabine Cornelsen}{University of Konstanz, Germany}{sabine.cornelsen@uni-konstanz.de}{https://orcid.org/0000-0002-1688-394X}{}

\author{Giordano {Da Lozzo}}{Roma Tre University,
  Italy}{}{https://orcid.org/0000-0003-2396-5174}{Supported, in part,
  by MUR of Italy (PRIN Project no.~2022ME9Z78~-- NextGRAAL and PRIN
  Project no.~2022TS4Y3N~-- EXPAND).}

\author{Luca Grilli}{Universit\`a degli Studi di Perugia,
  Italy}{}{https://orcid.org/0000-0002-2463-3772}{Supported, in part,
  by the Dipartimento di Ingegneria, Universit\`a degli Studi di
  Perugia through grant RICBA21LG.}

\author{Siddharth Gupta}{BITS Pilani, Goa Campus,
  India}{}{https://orcid.org/0000-0003-4671-9822}{Supported by the
  Engineering and Physical Sciences Research Council (EPSRC) through
  grant no.\ EP/V007793/1.}

\author{Jan Kratochv\'il}{Charles University, Prague, Czech
  Republic}{}{https://orcid.org/0000-0002-2620-6133}{Supported by the
  Czech Science Foundation through grant GA\v{C}R 23-04949X.}

\author{Alexander Wolff}{Universit\"at W\"urzburg, Germany}{}{https://orcid.org/0000-0001-5872-718X}{}

\authorrunning{S.~Cornelsen, G.~Da Lozzo, L.~Grilli, S.~Gupta,
  J.~Kratochv\'il, and A.~Wolff}

\Crefname{algorithm}{Algorithm}{Algorithms}
\DontPrintSemicolon
\SetArgSty{}
\SetKw{KwTo}{to}
\SetKw{KwOr}{or}
\SetKw{KwAnd}{and}
\SetKw{KwNot}{not}
\renewcommand{\emph}[1]{\textcolor{blue}{\em #1}}

\usetikzlibrary{calc}
\usepackage{framed}
\usepackage{tabularx}

\newlength{\RoundedBoxWidth}
\newsavebox{\GrayRoundedBox}
\newenvironment{GrayBox}[1]%
   {\setlength{\RoundedBoxWidth}{.93\columnwidth}
    \def\boxheading{#1}
    \begin{lrbox}{\GrayRoundedBox}
       \begin{minipage}{\RoundedBoxWidth}}%
   {   \end{minipage}
    \end{lrbox}
    \begin{center}
    \begin{tikzpicture}%
       \node(Text)[draw=black!20,fill=white,rounded corners,inner xsep=2ex,inner ysep=2ex,text width=\RoundedBoxWidth]
             {\usebox{\GrayRoundedBox}};
        \coordinate(x) at (current bounding box.north west);
        \node [draw=white,rectangle,inner sep=3pt,anchor=north west,fill=white]
        at ($(x)+(6pt,.75em)$) {\boxheading};
    \end{tikzpicture}
    \end{center}}

\newenvironment{defproblemx}[2]{\noindent\ignorespaces%
                                \FrameSep=6pt%
                                \parindent=0pt%
                 \vspace*{1ex}
                \begin{GrayBox}{#1}%
                \begin{tabular*}{\columnwidth}{!{\extracolsep{\fill}}@{\hspace{.1em}} >{\itshape} p{#2} p{0.84\columnwidth} @{}}%
            }{\\[-1.5ex]
                \end{tabular*}%
                \end{GrayBox}%
                \ignorespacesafterend
                \vspace*{1ex}
            }

\newcommand{\parametrizedProblem}[4]{%
  \begin{defproblemx}{#1}{1.6cm}
    Input: & #2 \\
    Parameter: & #3 \\
    Question: & #4
  \end{defproblemx}
}


\nolinenumbers %

\hideLIPIcs
  
\DeclareMathOperator{\seg}{seg}
\DeclareMathOperator{\lin}{line}


\relatedversion{A preliminary version of this paper appeared at GD
  2023~\cite{DBLP:conf/gd/CornelsenLGGKW23}.}

\acknowledgements{We thank the organizers of the Bertinoro Workshop on
  Graph Drawing 2023 where this research was initiated.}

\begin{document}
	
\maketitle
	
\begin{abstract}
  Given a straight-line drawing of a graph, a {\em segment} is a
  maximal set of edges that form a line segment.  Given a planar
  graph~$G$, the {\em segment number} of $G$ is the minimum number of
  segments that can be achieved by any planar straight-line drawing of
  $G$.  The {\em line cover number} of $G$ is the minimum number of
  lines that support all the edges of a planar straight-line drawing
  of $G$.  Computing the segment number or the line cover number of a
  planar graph is $\exists\mathbb{R}$-complete and, thus, NP-hard.
        
  We study the problem of computing the segment number from the
  perspective of parameterized complexity. We show that this problem is
  fixed-parameter tractable with respect to each of the following
  parameters: the vertex cover number, the segment number, and the
  line cover number.  We also consider colored versions of the
  segment and the line cover number.
\end{abstract}

\section{Introduction}

A \emph{segment} in a straight-line drawing of a graph is a
maximal set of edges that together form a line segment.  The
\emph{segment number} $\seg(G)$ of a planar
graph~$G$ is the minimum number of segments in any planar
straight-line drawing of~$G$~\cite{dujmovic_etal:compGeo07}.
The \emph{line cover number} $\lin(G)$ of a planar
graph~$G$ is the minimum number of lines that support all the edges of
a planar straight-line drawing of~$G$~\cite{chaplick_etal:CompGeo20}.
Clearly, $\lin(G) \le \seg(G)$ for any
graph~$G$~\cite{chaplick_etal:CompGeo20}.  As a side note, we show
that $\seg(G) \le \lin^2(G)$ for any connected graph~$G$.
For circular-arc drawings of planar graphs, the \emph{arc
  number}~\cite{s-dgfa-JGAA15} and \emph{circle cover
  number}~\cite{krw-dgfcf-JGAA19} are defined analogously as the
segment number and the line cover number, respectively, for straight-line
drawings.  For an example, see \cref{fig:numbers}.  All these numbers
have been considered as meaningful measures of the visual complexity of a drawing
of a graph; in particular, for the segment number, a user study has been conducted by Kindermann et al.~\cite{kindermannMS:JGAA2012}.

\begin{figure}[tb]
  \centering
  \includegraphics{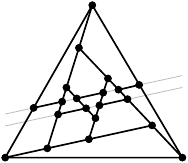} \hfil
  \includegraphics{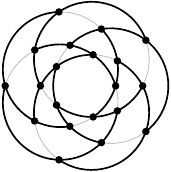}
  \caption{The dodecahedron graph has 30 edges, segment number~13,
    line cover number~10, arc number~10, and circle cover
    number~5~\cite{krw-dgfcf-JGAA19}.}
  \label{fig:numbers}
\end{figure}

In general, it is $\exists \mathbb{R}$-complete \cite{orw-ysng-GD19}
(and hence NP-hard) to compute the segment number of a planar graph.
(The complexity class $\exists \mathbb{R}$ was introduced by
Schaefer~\cite{Schaefer09}.)
The segment number can, however, be computed efficiently for
trees~\cite{dujmovic_etal:compGeo07},
series-parallel graphs of degree at most three~\cite{samee_etal:GD08},
subdivisions of outerplanar paths~\cite{Adnan2008}, 3-connected cubic
planar graphs~\cite{dujmovic_etal:compGeo07,igamberdievMS:jgaa17}, and
cacti~\cite{Goessmann_etal:WG22}.  Upper and lower bounds for the
segment number of various graph classes are known, such as outerplanar
graphs, 2-trees, planar 3-trees, 3-connected plane
graphs~\cite{dujmovic_etal:compGeo07}, (4-connected)
triangulations~\cite{durocher/mondal:compGeo2019}, and triconnected
planar 4-regular graphs~\cite{Goessmann_etal:WG22}.  Some of the lower
and upper bounds are so close that they yield constant-factor
approximations, e.g., for outerplanar paths, maximal outerplanar
graphs, 2-trees, and planar 3-trees~\cite{Goessmann_etal:WG22}.
Segment number and arc number have also been investigated under the
restriction that vertices must be placed on a polynomial-size grid
\cite{s-dgfa-JGAA15,kindermann_etal:GD19,hueltenschmidtKMS:jgaa18}.
Also in the setting that a planar graph~$G$ comes with a planar
embedding (that is, for each face, the cyclic ordering of the edges
around it is given and the outer face is fixed), it is NP-hard to
compute the segment number of~$G$ with respect to the given
embedding~\cite{durocherMNW:jgaa13}.

This paper focuses on the parametrized complexity of computing the
segment number of a graph. A decision problem with input $x$ and
parameter $k \in \mathbb{Z}_{\geq 0}$ is \emph{fixed-parameter
  tractable} (\emph{FPT}) if it can be solved by an algorithm with run 
time in $O(f(k)\cdot |x|^c)$ where $f$ is a computable
function, $|x|$ is the size of the input $x$ and $c$ is a constant. 
Given a planar graph $G$ and a parameter $k > 0$, the \textsc{Segment
  Number} problem is to decide whether $\seg(G) \leq k$. This is the
\emph{natural parameter} for the problem.  By considering additional
parameters, we get a more fine-grained picture of the complexity
of the problem.

Chaplick et al.~\cite{cflrvw-cdgfl-JGAA23} showed that
computing the line cover number of a planar graph is
in FPT with respect to its natural parameter.  They observed that
for a given graph~$G$ and an integer~$k$, the statement $\lin(G)\le k$
can be expressed by a first-order formula about the reals.  This
observation shows that the problem of deciding whether or not $\lin(G)
\le k$ lies in $\exists\mathbb{R}$: it reduces in polynomial time to
the decision problem for the existential theory of the reals.
The algorithm of Chaplick et al.\ crucially uses the
exponential-time decision procedure for the existential theory of the
reals by Renegar~\cite{Renegar92a,Renegar92b,Renegar92c} (see
\cref{sec:preliminaries}).  Unfortunately, this
procedure does {\em not} yield a geometric realization.  Chaplick et
al.\ even showed that constructing a drawing of a given planar graph
that is optimal with respect to the line cover number can be
unfeasible since there is a planar
graph~$G^\star$~\cite[Fig.~3b]{cflrvw-cdgfl-JGAA23} such that every optimal
drawing of~$G^\star$ requires irrational coordinates.  They argue that any
optimal drawing of~$G^\star$ contains the Perles configuration.  It is known
that every realization of the Perles configuration contains a point
with an irrational coordinate~\cite[page 23]{Berger2010}.  Moreover,
any optimal drawing of~$G^\star$ admits a cover with ten lines, and each of
these lines contains a {\em single} line segment of the drawing.  In other
words, every drawing of~$G^\star$ that is optimal with
respect to the line cover number is also optimal with respect to the
segment number.

\subparagraph{Our results}

We show that computing the segment number of
a graph is FPT with respect to each of the following parameters:
the natural parameter, the line cover number (both in
\cref{sec:segment-number}), and the \emph{vertex cover number} (in
\cref{sec:vertex-cover-number}).  Recall that the vertex cover number
is the minimum number of vertices that have to be removed such that
the remaining graph is an independent set.
In the parametrized complexity community, the vertex cover number is
considered a rather weak graph parameter, but for (geometric) graph
drawing problems, FPT results can be challenging to obtain even
with respect to the vertex cover number
\cite{bgmn-pabep-JGAA20,bgmn-paql-JGAA22,bcghvw-bcong-ESA22}.

We remark that our algorithms use Renegar's decision procedure as a
subroutine; hence, when we compute the segment number~$k$ of a planar
graph~$G$, we do {\em not} get a straight-line drawing of~$G$ that
consists of $k$ line segments.  Recall, however, that even {\em
  specifying} such a drawing is difficult for some graphs such as the
above-mentioned graph~$G^\star$, which has a point with irrational
coordinates in any drawing that is optimal with respect to the segment
number.

Motivated by list coloring, we also consider colored versions of the
segment and line cover number problems.
As a warm up, we provide efficient algorithms for computing the
segment number of \emph{banana trees} and \emph{banana cycles}, that
is, graphs that can be obtained from a tree or a cycle, respectively,
by replacing each edge by a set of paths of length two; see
\cref{sec:banana-graphs}.

\section{Preliminaries}
\label{sec:preliminaries}

In a straight-line drawing of a graph, two incident edges are \emph{aligned} if they are on the same segment.  Since the number of segments equals the number of edges minus the number of pairs of aligned edges,
we get the following. 

\begin{lemma}\label{LEMMA:max-aligned}
	A planar straight-line drawing has the minimum number of segments if and only if it has the
	maximum number of pairs of aligned edges.
\end{lemma}

An \emph{existential first-order formula about the reals} is a formula
of the form $\exists x_1\dots\exists x_m$ $\phi(x_1,\dots,x_m)$, where
$\phi$ consists of Boolean combinations of order and equality relations
between polynomials with rational coefficients over the variables
$x_1,\dots,x_m$.
Renegar's result on the existential theory of the reals can be
summarized as follows.

\begin{theorem}[Renegar~\cite{Renegar92a,Renegar92b,Renegar92c}]
  \label{thm:renegar}
  Given any existential sentence $\Phi$ about the reals, one can
  decide whether $\Phi$ is true or false in time
  \[
    (L\log L\log\log L) \cdot (PD)^{O(N)},
  \]
  where~$N$ is the number of variables in~$\Phi$,~$P$ is the
  number of polynomials in~$\Phi$, $D$ is the maximum total
  degree over the polynomials in~$\Phi$, and~$L$ is the maximum length
  of the binary representation over the coefficients of the polynomials
  in~$\Phi$.
\end{theorem}

The proof of the next lemma follows the approach by Chaplick et
al.~\cite[Lemma~2.2]{cflrvw-cdgfl-JGAA23}.

\begin{lemma}\label{lem:firstOrder}
	Given an $n$-vertex planar graph $G$ and an integer $k \le 3n-6$,
	there exists a first-order formula~$\Phi$ about the reals that
	involves $O(n^2)$ polynomials in $O(n)$ variables (each of constant
	total degree and with integer coefficients of constant absolute
	value) such that $\Phi$ is satisfiable if and only if
	$\seg(G) \leq k$.
\end{lemma}

\begin{proof}
	Note that we can return true immediately if $k > 3n-6$.  Otherwise,
	the formula $\Phi$ needs to model the fact that there are $k$ pairs
	of points, determining a set $\mathcal{L}$ of $k$ lines, and that
	there are $n$
	distinct points representing the vertices of $G$ such that the segments
	corresponding to the edges of $G$ lie on the lines in~$\mathcal{L}$
	and do not overlap or cross each other or leave gaps on the
	lines.  For details on how to encode this, see
	\cite[Lemma~2.2]{cflrvw-cdgfl-JGAA23}.
\end{proof}

\cref{lem:firstOrder,thm:renegar} immediately imply the following.

\begin{corollary}\label{lem:expTimeAlgo}
  Given an $n$-vertex planar graph $G$ and an integer $k$, there exists a
  $2^{O(n)}$-time algorithm to decide whether $\seg(G) \le k$.
\end{corollary}

In the remainder of this paper, for a graph $G$, we write $V(G)$ for
the vertex set of~$G$ and $E(G)$ for the edge set of~$G$.

\section{The Segment Number of Banana Trees and Cycles}
\label{sec:banana-graphs}

We first consider some introductory examples.  A \emph{banana} is the
union of paths that share only their start- and endpoints
\cite{ss-isglc-JCTSB20}.  In this paper, we additionally insist that
all paths have length~2.  We say that a banana is a
\emph{$k$-banana} if it consists of $k$ paths (of length~2).  We call
the joint endpoints of the paths \emph{covering vertices} and the
vertices in the interior of the paths \emph{independent vertices}.
\emph{Banana graphs}, i.e., graphs in which some edges are replaced by bananas will play an important role in \cref{sec:vertex-cover-number} when we consider the parameter vertex cover.

\begin{lemma}[\cite{dujmovic_etal:compGeo07}]
	A $k$-banana has segment number $\lfloor3k/2\rfloor$.
\end{lemma}
\begin{proof}
	A $k$-banana has $2k$ edges. There is a drawing in which
        $\lceil k/2 \rceil$ pairs of edges are aligned; see
        \cref{SUBFIG:9-banana}. This is optimum: At most one
	path of length two can be aligned at the independent vertex.
        If a  pair of edges is aligned at one
	covering vertex then no pair 
	can be aligned at the other covering vertex.
\end{proof}

\begin{figure}[tb]
  \centering
  \begin{subfigure}[t]{.3\linewidth}
	\centering
	\includegraphics[page=2]{banana-path}
	\caption{a 9-banana}
	\label{SUBFIG:9-banana}
  \end{subfigure}\hfill
  \begin{subfigure}[t]{.6\linewidth}
	\centering
	\includegraphics[page=1]{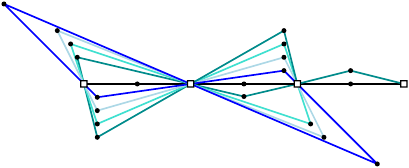}
	\caption{a banana-path of length~3}
	\label{SUBFIG:banana-path}
  \end{subfigure}
  \caption{Banana graphs}
  \label{FIG:bananaPath}
\end{figure}

Given an integer $\ell>0$, a \emph{banana path of length $\ell$} is a
graph that is constructed from a path $\left<v_0,\dots,v_\ell\right>$
of length~$\ell$ by replacing, for $i \in \{1,\dots,\ell\}$,
edge~$\{v_{i-1},v_i\}$ of the path by a $k_i$-banana, called
banana~$i$, for some value~$k_i>0$.  We say that $k_i$ is the
\emph{multiplicity} of edge~$\{v_{i-1},v_i\}$; see
\cref{SUBFIG:banana-path}.  \emph{Banana trees} and \emph{banana
  cycles} are defined analogously; see \cref{FIG:bananaCycle}.

\begin{theorem}\label{THEO:bananaPath}
  The segment number of a banana path of length~$\ell$ can be computed
  in $O(\ell)$ time and can be expressed explicitly as a function of
  the multiplicities.
\end{theorem}
\begin{proof}
  Let $B$ be a banana path of length~$\ell$ with multiplicities
  $k_1,\dots,k_\ell$.  In each banana, align two edges at an
  independent vertex.  At each inner covering
  vertex $v_i$ with $i \in \{1,\dots,\ell-1\}$, align
  $a_i=\min\{k_i,k_{i+1}\}$ edges from bananas~$i$ and $(i+1)$.
  In particular, we align the edges from incident bananas that were
  aligned at their respective independent vertex. 
  Further, let $a_0 = a_\ell = 1$. For each $i \in \{1,\dots,\ell\}$,
  let $s_i=\max\{k_i-a_{i-1}, k_i-a_i\}$.
   Align $\lfloor s_i/2 \rfloor$ pairs of edges of banana~$i$ at covering vertex $v_{i-1}$ or $v_i$, based on where the maximum is assumed. See \cref{SUBFIG:banana-path}.
  Note that setting $a_0 = a_{\ell}=1$ takes into account that one
  edge is already aligned at an independent vertex.  Therefore, it
  cannot be aligned at a covering vertex at the same time.
	It is not possible to align more pairs of edges. 
	Thus, the segment number is 
  $\sum_{i=1}^\ell (2k_i - 1 - a_i - \lfloor s_i/2 \rfloor)$.
\end{proof}

\begin{figure}[tb]
	\begin{subfigure}[b]{.5\linewidth}
		\centering
		\includegraphics[page=6]{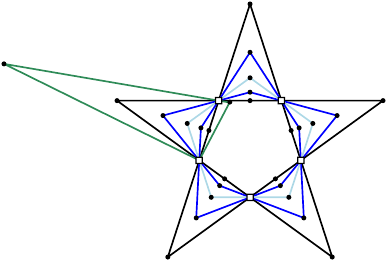}
		\caption{odd length cycle}
		\label{FIG:oddcycle}
	\end{subfigure}
	\hspace{-1ex}
	\begin{subfigure}[b]{.5\linewidth}
		\centering
		\includegraphics[page=2]{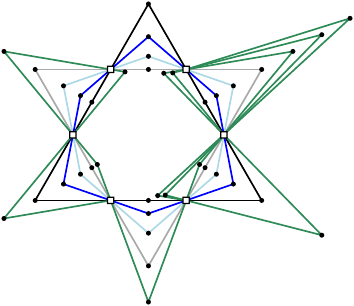}
		\caption{even length cycle}
		\label{FIG:evencycle}
	\end{subfigure}
	\caption{Segment minimum drawings of banana cycles.}
	\label{FIG:bananaCycle}
\end{figure}

Observe that the strategy of locally optimal alignments would no
longer work for simple cycles.  E.g., assume each edge in a simple
cycle was replaced by a 1-banana.  Then locally, all pairs of incident
edges should be aligned, both at vertices of the simple cycle and at
the independent vertices.  But this is not realizable geometrically.
Instead the segment number would be~3.  Some examples of
segment-minimum drawings of cycles of 4-bananas are shown in
\cref{FIG:bananaCycle}.

\begin{theorem}\label{THEO:bananaTreesCycles} The segment number of
	\begin{enumerate}
		\item 
	 a banana tree and
\item
		a banana cycle of length at least five where each banana contains at least two independent vertices
	\end{enumerate}
can be computed in time linear in the number of covering vertices.
\end{theorem}

\begin{proof}
	\textbf{Banana tree:}
	Let $G$ be a banana tree constructed from a tree~$T$. For an edge~$e=\{u,v\}$ of~$T$, let $k_{uv} = k_{vu}$ be the  multiplicity of $e$ and let $E_{uv}$ be the set of edges incident to $v$ that are contained in the $k_{uv}$-banana inserted for $e$.
	For each banana, align the two edges incident to one independent
	vertex.  For each leaf~$w$ of~$T$, align the remaining edges of the
	respective banana~--~up to at most two~--~at~$w$.  Let $v$ be an inner vertex of $T$, and
	let $v_1,\dots,v_\ell$ be the neighbors of~$v$ in~$T$.  Align as
	many pairs of edges from different bananas at~$v$.  These are all (except possibly one)
	edges if 
	there is no $j \in\{1,\dots,\ell\} $ such that
	$k_{vv_j} > \sum_{i \neq j} k_{vv_i}$.
	If there is a $k_{vv_j}$-banana that is \emph{large} at $v$, i.e., if $k_{vv_j} > \sum_{i \neq j} k_{vv_i}$ for some~$j$ then align 
	some edges inside the respective $k_{vv_j}$-banana at~$v$ -- or
	at~$v_j$ if the $k_{vv_j}$-banana is also large at $v_j$ and the alignment at this
	end would be more beneficial. More precisely: let $u_1,\dots,u_k$ be the neighbors of $v_j$ in $T$ different from $v$.
	Then align $\lfloor(k_{vv_j} - \sum_{i \neq j} k_{vv_i})/2\rfloor$ pairs of edges within $E_{v_jv}$ at $v$ or $\lfloor(k_{vv_j} - \sum_{i=1}^k k_{v_ju_i})/2\rfloor$ pairs of edges within $E_{vv_j}$ at $v_j$, depending on which number is larger.
	Observe that at any vertex this
	condition always holds for at most one banana, so there are no
	conflicts. Locally, this yields an optimal assignment. 
	
	\begin{figure}
		\centering
		\subcaptionbox{large banana at $v$ towards the root\label{SUBFIG:treeB}}{\includegraphics[page=1]{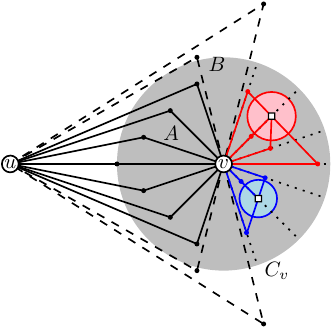}} \hfil
		\subcaptionbox{large banana at $v$ towards a subtree\label{SUBFIG:treenoB}}{\includegraphics[page=2]{banana-tree}}
		\caption{How to align the edges at an inner vertex $v$ of a banana tree.}
		\label{FIG:banana-tree}
	\end{figure}

	The respective alignment can be realized geometrically: We root $T$ at a degree one vetex $r$. All but at most two edges incident to $r$ must be aligned at $r$ and two edges incident to one independent vertex incident to $r$ must also be aligned. This can be realized as in \cref{SUBFIG:9-banana}. In general, we have now the following situation, we have drawn the bananas contained in a subtree $T'$ of $T$.  All vertices of $T'$ that are incident to vertices not yet in $T'$ are leaves of $T'$. Let $v$ be such a leaf of $T'$ and let $u$ be the ancestor of $v$ in $T'$. Let $A \subseteq E_{uv}$ be the set of edges that should be aligned with edges in other bananas and let $B\subseteq E_{uv}$ be the edges that are aligned with  each other. We maintain the property that the maximum angle between two edges of $A$ is less than $\pi$. Thus, if $B\neq \emptyset$ then $B$ encloses $A$. See \cref{SUBFIG:treeB}. Moreover, for each leaf $v$ of $T'$, we maintain disjoint circles $C_v$ around $v$ that contain or intersect at most edges from $E_{uv} \cup E_{vu}$. We now insert the subtree rooted at $v$ into these circles $C_v$ maintaining the alignment requirements, the disjoint circles and the above mentioned drawing requirements. Observe that if $B=\emptyset$ then there might be a different large banana at $v$ (the red banana in \cref{SUBFIG:treenoB}). Since the maximum angle at $v$ between two edges in $A$ is less than $\pi$ it follows that we can always fit this large banana.

	\textbf{Banana cycle:}
	Let $G$ be a banana cycle of length at least five with the property that each banana contains at least two independent vertices.
	Let $H$ be a maximal subgraph of $G$ such that $H$ is a regular
	banana cycle, i.e., such that each banana contains the same number
	of independent vertices.  In \cref{FIG:evencycle} this is the
	subgraph induced by the black, gray, blue, and lightblue edges. Depending on whether the length of the cycle is odd or even, draw
	$H$ as indicated in \cref{FIG:bananaCycle} for banana cycles of
	length~5 or~6 starting with the one or two cycles, respectively,
	containing pairs of edges aligned at independent vertices (black and
	gray).  Removing the edges of $H$ from $G$ yields a union of banana
	paths.  We align edges according to the rules for banana paths without
	aligning edges at independent vertices (edges in
	different shades of green).
\end{proof}

\section{Parameter: Vertex Cover Number}
\label{sec:vertex-cover-number}
In this section, we consider the following parametrized
problem.
\parametrizedProblem{\sc Segment Number by Vertex Cover Number}%
{A simple planar graph $G$, an integer~$s$.}%
{Vertex cover number $k$ of $G$.}%
{Is the segment number of $G$ at most~$s$?}
Let $G$ be a simple planar graph, and let $V' \subset V(G)$ be a
vertex cover of~$G$ with $k$ vertices.  In order to compute the
segment number of $G$, we first compute a subgraph of $G$ whose size
is in $O(2^k)$. We vary over all possible alignments of this
subgraph and check whether they are geometrically realizable.
We finally use an integer linear program (ILP) in order to reinsert
the missing parts from $G$.  In the end, we take the best among the
thus computed solutions.  The details are as follows.

\subparagraph{Dividing $V \setminus V'$ into equivalence classes}

Two vertices of $v \in V(G) \setminus V'$ are \emph{equivalent} if and
only if they are adjacent to the same set of vertices in~$V'$. 
We say that an equivalence class $C$ is a \emph{$j$-class} if
each vertex $v \in C$ is adjacent to exactly $j$ vertices in
$V'$.  Observe that the $j$-classes contain at most two
vertices if $j > 2$. Otherwise $G$ would contain a
$K_{3,3}$, contradicting that $G$ is planar. Thus, the
number of vertices in $j$-classes, $j>2$ is bounded by
$2 \cdot \sum_{j=3}^k{k \choose j} \in O(2^k)$.

\subparagraph{Reducing the size of the graph}

Let $G_1$ be the graph obtained from $G$ by removing all
vertices contained in 1-classes.  Consider a planar embedding of~$G_1$.
Let $C$ be a 2-class and let $v,w \in V'$ be adjacent to all vertices in~$C$.  
Observe that the edges between $v$ and $C$ do not have to be
consecutive in the cyclic order around~$v$ (and similarly
for~$w$).  A \emph{contiguous 2-class} is a maximal subset
of~$C$ such that the incident edges are consecutive around
both~$v$ and~$w$.  Note that two contiguous 2-classes are 
separated by the edge $\{v,w\}$ or by at least a
vertex-cover vertex.  Thus, a
2-class consists of at most $k$ contiguous 2-classes.  Now,
for each 2-class~$C$, we remove all but $\min(|C|, k)$
vertices from~$C$.  Let the resulting graph be~$G_2$. 
Observe that the vertices of $G_2$ are the vertex cover vertices, 
the vertices of all $j$-classes, $j>2$, and at most $k$ vertices from each 2-class.
Thus, the number of vertices of $G_2$ is bounded by 
$k + 2 \cdot \sum_{j=3}^k{k \choose j} + k \cdot {k \choose 2} \in O(2^k)$.

\subparagraph{Alignment within 2-classes}

For each 2-class $C$, let $e_C$ be the edge connecting 
the two vertex cover vertices adjacent to the vertices in $C$. 
We vary over all subsets $\mathcal C$ of 2-classes $C$ 
for which $e_C$ is not an edge of $G_2$.
There are at most $2^{3k}$ such subsets:
Observe that the planar graph induced by $V'$ and one vertex from each
2-class is a subdivision of a graph with vertex set $V'$ and edge set
$\{e_C \colon C \text{ 2-class}\}$, which implies that the number of
2-classes is at most $3k-6$.
For each $C \in \mathcal C$, we add $e_C$ to $G_2$. 
These edges represent  a pair of aligned edges at an independent vertex. 
Let the resulting graph be $G_2(\mathcal C)$. 

\subparagraph{Varying over all planar embeddings $\mathcal E$
  of~$G_2(\mathcal C)$}
\label{PAGE:embeddings}

For each 2-class $C$, we replace each contiguous 2-class
in~$\mathcal{E}$ by four vertices.
The resulting four 2-paths represent the boundaries of two consecutive 
contiguous 2-classes; each of which shall be drawn in a non-convex way, i.e., 
such that the segment between the two incident vertex cover vertices
does not lie in its interior.
We call these non-convex quadrilaterals \emph{boomerangs}.
Let the thus constructed plane graph be $G_2(\mathcal C,\mathcal E)$.        
          
The reason why we represent each contiguous 2-class by two non-convex quadrilaterals instead of one arbitrary quadrilateral will become clear
when we reinsert the 2-classes into the boomerangs; see also \cref{SUBFIG:insertingBoomerangs}: 
Let $u \in C$, let $v$ and $w$ be the two vertex cover vertices
incident to $u$ and assume that we have opted to draw $u$ into boomerang~$b$.
Then we need to ensure that the edges $e_v = \{v,u\}$ and $e_w=
\{w,u\}$ meet inside $b$ when we draw~$e_v$ starting at $v$ and~$e_w$
starting at $w$ independently of each other and with arbitrary slopes.

\begin{figure}
	\centering
	\subcaptionbox{splitting boomerangs\label{FIG:banana-split}}[0.3\linewidth]{\centering \includegraphics[page=2]{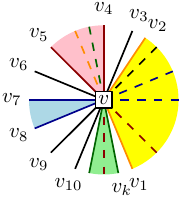}} 
	\hfil
	\subcaptionbox{reinserting 2-classes into a boomerang\label{SUBFIG:insertingBoomerangs}}[0.5\linewidth]{\centering \includegraphics[page=3]{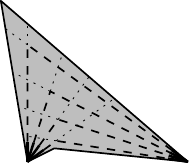}}
	\caption{\label{FIG:splitAndInsert} We vary over all pairings
		$\pi$ of edges and boomerangs around each
		vertex.  (a)~The dashed edges are inserted into
		$G_2(\mathcal C,\mathcal E)$ in order to represent alignment
		requirements with boomerangs.  If the alignment requirements of
		the thus constructed graph $G_2(\mathcal C,\mathcal E,\pi)$ can
		be geometrically realized, we use an ILP to distribute the
		remaining edges.  (b)~How to geometrically insert the 2-classes
		into a boomerang.}
\end{figure}

\subparagraph{Alignment Requirements}

We vary over all possible ``pairings" $\pi$ between edges and/or
boomerangs incident to a common vertex $v$ that respect the
embedding~$\mathcal E$.  If we require that two edges are aligned,
that an edge should be aligned with an edge inside a boomerang, or
that two edges inside different boomerangs can be aligned, we say that
the respective edges and boomerangs are \emph{paired}.  E.g.,
\cref{FIG:banana-split} illustrates the following pairing around a
vertex $v$: the yellow boomerang is paired with the blue and the red
boomerang as well as with the edge $\{v,v_9\}$.  The red and the green
boomerangs are also paired.  Edges $\{v,v_{10}\}$
and $\{v,v_{3}\}$ are paired.  Edge $\{v_6,v\}$ is not paired with any other edge.

In order to handle pairings within boomerangs, we
insert further edges into $G_2(\mathcal C,\mathcal E)$.
See the dashed edges in \cref{FIG:banana-split}.
Let $v$ be a vertex cover
vertex. Let $v_1,\dots,v_k$ be the neighbors of $v$ in this
order around $v$. Assume that $v_1$ and $v_2$ are two
independent vertices representing the border of a
boomerang~$b$. 
Assume that we would like to align the
edges around $v$ such that $\{v,v_i\},\dots,\{v,v_j\}$
should get aligned with some edges in $b$ if
any. Then we add up to $j-i+1$ edges incident to $v$ inside the region representing $b$ and require them to be
aligned with the respective edges among  $\{v,v_i\},\dots,\{v,v_j\}$ at $v$.      

E.g., when we consider the yellow boomerang $b$ in \cref{FIG:banana-split} then $i=5$, $j=9$. Since we require the red boomerang to be partially paired with $b$ (we also want it to be paired with the green one), we only add a counterpart for one of the boundary edges, $\{v,v_5\}$, inside $b$. The blue boomerang should be exclusively paired with $b$, so we add counter parts for both boundary edges, $\{v,v_7\}$ and $\{v,v_8\}$. Edge $\{v,v_9\}$ should be aligned with an edge in $b$, so it also gets a counterpart in $b$. According to the currently chosen pairing $\pi$, the edge $\{v,v_6\}$ does not have to be aligned in $v$, so we do not add a counterpart for it in~$b$. 

We do this operation for each boomerang and
each incident vertex cover vertex.  The resulting graph $G_2(\mathcal C,\mathcal E,\pi)$ has at most thrice as many edges~as~$G_2(\mathcal C,\mathcal E)$.

\begin{figure}[tb]
	\centering
	\includegraphics{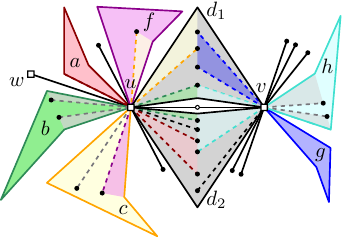}
	\caption{Example situation for $G_2(\mathcal C,\mathcal E,\pi)$ as a basis for the description of the~ILP.}
	\label{FIG:ilp-example}
\end{figure}

\subparagraph{Testing and optimally reinserting the 1- and 2-class vertices}

We now have a plane graph $G_2(\mathcal C,\mathcal E,\pi)$ with
$O(2^k)$ vertices. See \cref{FIG:ilp-example}. Some pairs of
edges are required to be aligned. The two 2-paths representing the
boundary of a boomerang must bound a quadrilaterial that does not
contain the segment between the two vertex cover vertices.  We use
Renegar's algorithm \cite{Renegar92a,Renegar92b} to test
whether these requirements are geometrically realizable. 

\begin{lemma}
  \label{lem:realize-G-two}
  Given a plane graph $H$ with embedding~$\mathcal E$ and $K$
  vertices, a set~$A$ of pairs of edges, a set~$B$ of 4-cycles
  of~$H$, and for each 4-cycle~$Q$ in~$B$ a specified diagonal
  $d(Q)$, we can decide, in $K^{O(K)}$ time, whether there
  exists a planar straight-line drawing of~$H$ such that (i)~the edge
  pairs in~$A$ are aligned and (ii)~for each 4-cycle~$Q$ in~$B$,
  $d(Q)$ does not lie in the interior of~$Q$ or on~$Q$.
\end{lemma}

\begin{proof}
	We use the following simple tool from geometry.  Given three points
	$a=(x_1,y_1)$, $b=(x_2,y_2)$, $c=(x_3,y_3)$ in the plane, let
	\[
	\chi(a,b,c)=\left|
	\begin{array}{ccc}
	x_1&y_1&1\\
	x_2&y_2&1\\
	x_3&y_3&1
	\end{array}
	\right|
	\]
	be the scalar triple product of three 3-dimensional vectors
	$(x_1,y_1,1)$, $(x_2,y_2,1)$, and $(x_3,y_3,1)$.  As is well
        known,
	the following conditions are equivalent:
	\begin{itemize}
		\item $\chi(a,b,c)>0$;
		\item $a\ne b$, and the point $c$ lies in the left halfplane with
		respect to the oriented line through~$a$ and~$b$, oriented
		from~$a$ to~$b$;
		\item the points $a$, $b$, and $c$ are pairwise distinct,
		non-collinear, and occur counterclockwise in the circumcircle of
		the triangle~$\triangle abc$ in this order.
	\end{itemize}
	Moreover, $\chi(a,b,c)=0$ if and only if the points $a$, $b$, and
	$c$ are collinear, including the case that some of them coincide.
	
	Let $\{v_1,\dots,v_K\}$ be the vertex set of~$H$.  We want to
	express the fact that there are $K$ pairwise distinct points
	$(x_1,y_1),\dots,(x_K,y_K)$ that determine a straight-line drawing
	of~$H$.  Our existential statement~$\Phi$ about the reals begins,
	therefore, with the quantifier prefix	
	\[\exists x_1 \exists y_1 \dots \exists x_K \exists y_K
	\bigwedge_{1 \le i < j \le K} \neg (x_i=x_j \land y_i=y_j).\]
	For each pair of edges $e=\{u,v\}$ and $f=\{v,w\}$ such that~$f$
	follows~$e$ immediately in the counterclockwise cyclic ordering of
	the edges incident to~$v$ according to~$\mathcal{E}$, we add
	to~$\Phi$ the constraint $\chi(u,w,v)>0$ (using conjugation).
	
	For each 4-cycle $Q=\langle a,b,c,d \rangle \in B$ with $d(Q)=\langle a,c \rangle $, we add
	to~$\Phi$ the constraint $\chi(a,b,c)\cdot\chi(a,d,c) > 0$.
	In other words, $\chi(a,b,c)$ and $\chi(a,d,c)$ must have the
	same (non-zero) sign.
	
	If there is a pair of edges in $A$ that is not incident, we can reject the instance.
	For each pair of edges $e=\{u,v\}$ and $f=\{v,w\}$ in~$A$, we add
	to~$\Phi$ the constraint $\chi(u,v,w)=0$.
	(We could also express that $v$
	must lie in the relative interior of the line segment
	$\overline{uw}$, but this is implied automatically by the ordering
	of the edges around~$v$~-- if there is another edge incident to~$v$.
	If not, we can remove~$v$ and replace the path
	$\langle u,v,w \rangle$ by the straight-line edge $\{u,w\}$.)
	
	For each pair of non-incident edges $e=\{u,v\}$ and
	$e'=\{u',v'\}$ of~$H$, we need to ensure that~$e$ and~$e'$ do
	not cross each other.  This can be expressed using the
	relation $D(u,v,u',v')$ defined in the proof of
	\cite[Lemma~2.2]{cflrvw-cdgfl-JGAA23}.
	
	Note that~$\Phi$ involves $O(K^2)$ polynomials in
	$2K$ variables, each of constant total degree and with integer
	coefficients of constant absolute value.  Using Renegar's result
	(\cref{thm:renegar}), we can test whether~$\Phi$ admits a solution
	over the reals in $K^{O(K)}$ time.
\end{proof}

We apply \cref{lem:realize-G-two} to
$H=G_2(\mathcal C,\mathcal E,\pi)$, setting $A$ according to~$\pi$,
$B$ to the set of all boomerangs,
and, for each boomerang~$b$, the specified diagonal $d(b)$ to the
segment between the two vertex cover vertices.
If the answer is yes, we set up and solve an ILP to
optimally insert 2-class vertices into the boomerangs and to optimally
add the 1-class vertices.  To this end, we use, among others, the
following variables: $x_{v,b,d}=x_{v,d,b}$ expresses how many edges in
boomerang $b$ should be aligned to edges in boomerang $d$ at vertex
$v$ and $y_{v,d}$ expresses how many edges between $v$ and 1-class
vertices should be aligned with edges in boomerang~$d$.  We allow that
some of the boomerangs remain empty. 
The details of the ILP are as follows.

\subparagraph{Description of the ILP}

We explain the ILP for the optimal distribution of 2-class and 1-class
vertices by example.  Let $d_1$ and $d_2$ be the two gray boomerangs 
in \cref{FIG:ilp-example} with vertex cover vertices $u$ and $v$, let $a,b,c$ be
the boomerangs that are incident to~$u$ and paired with~$d$, and let~$g$
and~$h$ be the boomerangs incident to~$v$ and paired with~$d$.
(Additionally, boomerang~$f$ is paired with boomerang~$c$.)

With the respective upper case letters we will denote the respective
2-classes. E.g., $D$ is the 2-class whose vertices are adjacent to~$u$
and~$v$ since the boomerangs $d_1$ and $d_2$ represent contiguous
subsets of~$D$.
For a 2-class $C$, let~$\mu_{C}$ be the size of $C$ --  minus one if
$C \in \mathcal C$, i.e., if one 2-path through a vertex of $C$ was
already modeled by the edge~$e_C$. 
For the ILP, $\mu_C$ is a constant.  

Let now $b$ and $d$ be two boomerangs that are paired at the vertex cover
vertex~$v$.  We introduce an integer variable $x_{v,b,d}=x_{v,d,b}$
for the number of pairs $\{e,e'\}$ of edges such that $e$ is in~$b$,
$e'$ is in~$d$, and $e$ and $e'$ are aligned at~$v$.

For each vertex~$v$, let $\lambda_v$ be the number of leaves
(1-classes) incident to~$v$.  For the ILP, $\lambda_v$ is a constant.
Further, we introduce an integer variable $y_{v,d}$ for the number of
leaves incident to~$v$ that are aligned with edges assigned to a
boomerang~$d$, and a variable $y_v$ for the number of pairs of
leaves incident to~$v$ that are aligned with each other.  Then,
$\lambda_v$ is an upper bound to the sum of all leaves that are
aligned with edges of any boomerang incident to~$v$ plus twice the number
of pairs of leaves that are aligned with each other.  For vertices~$u$
and~$v$ in \cref{FIG:ilp-example}, this yields the constraints
\begin{align*}
y_{v,d_1} + y_{v,d_2} + y_{v,g} + y_{v,h} + 2y_v &\le \lambda_v, \\
y_{u,a} + y_{u,b} + y_{u,c} + y_{u,d_1} + y_{u,d_2} + y_{u,f} + 2y_u &\le \lambda_u.
\end{align*}
For each vertex~$v$, let~$x_v$ be defined as follows:
\[x_v = \sum_{b,d \text{ boomerangs incident to } v} x_{v,b,d} \quad +
\sum_{d \text{ boomerang incident to } v} y_{v,d}.\]
For example, in \cref{FIG:ilp-example}, we have 
\begin{align*}
x_u &= (x_{u,a,d_2} + x_{u,b,d_1} + x_{u,b,d_2} + x_{u,c,d_1} +
x_{u,c,f}) + (y_{u,a} + \dots + y_{u,f}) , \\
x_v &= (x_{v,d_1,g} + x_{v,d_1,h} + x_{v,d_2,h}) +
(y_{v,d_1} + y_{v,d_2} + y_{v,g} + y_{v,h}) .
\end{align*}
Then our objective is to
\[ \text{maximize } \sum_{v \in V(G)} (x_v + y_v). \]
For a boomerang~$b$ incident to vertex cover
vertex~$v$, let $\varepsilon_{v,b}$ be the number of edges of~$G_2$
that are paired at~$v$ with edges of~$b$.  For the ILP,
$\varepsilon_{v,b}$ is a constant.  For example, in
\cref{FIG:ilp-example}, $\varepsilon_{u,d_2}=1$ due to the edge~$wu$
in~$G_2$, whereas $\varepsilon_{u,d_1}=0$.
We need only one more type of constraints.
We establish, for each boomerang $d_i$ representing a contiguous
subclass of the 2-class~$D$, a bound~$z_{d_i}$ for the number
of vertices that have to be contained in~$d_i$ in order to satisfy the
alignment requirements.
\begin{align*}
x_{u,b,d_1} + x_{u,c,d_1} + y_{u,d_1} + \varepsilon_{u,d_1} &\le z_{d_1} \\
x_{v,g,d_1} + x_{v,h,d_1} + y_{v,d_1} + \varepsilon_{v,d_1} &\le z_{d_1} \\
x_{u,a,d_2} + x_{u,b,d_2} + y_{u,d_2} + \varepsilon_{u,d_2} &\le z_{d_2} \\
x_{v,h,d_1}+ y_{v,d_2}  + \varepsilon_{v,d_2} &\le z_{d_2}
\end{align*}
We also do this for all other 2-classes.
Now the following constraint says that the total number of vertices
that we distribute to the boomerangs associated with the 2-class~$D$ is
at most~$\mu_D$.  In our example, we have
\[z_{d_1} + z_{d_2} \le \mu_D \]
Note that the edges that belong to aligned pairs coming into a
boomerang from one side can always be completed to paths of length~2,
either with edges that belong to aligned pairs entering the
boomerang from the other side or with usual unconstrained edges.

It remains to analyze the numbers of variables and constraints.
Recall that the number of vertices and edges of the plane graph
$G_2(\mathcal C,\mathcal E,\pi)$ is in $O(2^k)$.  The number
of variables of type~$y_{v,d}$ and~$z_d$ is bounded by the number of
edges and, thus, is also in $O(2^k)$.  The number of
variables of type~$x_v$ and~$y_v$ is obviously in $O(2^k)$.
Finally, for a fixed vertex cover vertex $v$, the number of variables
of type~$x_{v,b,d}$ is bounded by $\deg v$: If $b$ is only paired
with~$d$, then we assign the cost for the variables
$x_{v,b,d}=x_{v,d,b}$ to~$b$.  If both, $b$ and $d$ are not only
paired with each other, we assign the cost of $x_{v,b,d}$ to~$b$ and
the cost of $x_{v,d,b}$ to~$d$.  Observe that this happens at most
twice for each boomerang incident to~$v$.  Thus, by the hand shaking
lemma, the total number of variables of type~$x_{v,b,d}$ is bounded by
the number of edges, which is again in $O(2^k)$.
There is one constraint for each 1-class, one constraint for each
variable of type~$x_v$, two constraints for each variable of
type~$z_d$, and one constraint for each 2-class. Thus, the total
number of constraints is in $O(2^k)$.

Summarizing the above discussion, the ILP uses $O(2^k)$
variables and constraints and, thus, can be solved in $2^{O(k2^k)}$ time
\cite{DBLP:journals/mor/Kannan87,DBLP:journals/mor/Lenstra83,DBLP:journals/combinatorica/FrankT87}.

We are now ready to prove the main theorem of this section.

\begin{theorem}\label{THEO:vertex-cover}
  \textsc{Segment Number by Vertex Cover Number} is FPT.
\end{theorem}

\begin{proof}
	The vertex cover number of the input graph $G$ with $n$ vertices
	can be computed in $O(kn+1.274^{k})$ time
	\cite{chenKanjXia:tcs10}.
	The number of subsets $\mathcal C$ of 2-classes is in
        $O(2^{3k})$. Then we vary over all embeddings of $G_2(\mathcal
        C)$; see the respective paragraph on
        \cpageref{PAGE:embeddings}.
	The number of embeddings $\mathcal E$
        of the graph $G_2(\mathcal C)$ with $O(2^k)$ vertices is in
        $O((c2^k)!)$ for some constant $c$, and
	the number of possible pairings $\pi$ is also bounded by a function of $k$. 
	The plane graph $G_2(\mathcal C, \mathcal E,\pi)$ has $O(2^k)$ vertices.
	Hence, testing whether the alignment and non-convexity requirements of $G_2(\mathcal C, \mathcal E,\pi)$ can be realized geometrically can be done in $2^{O(k2^k)}$ time. 
	If the answer is yes, we can solve the ILP in $2^{O(k2^k)}$ time. 
	In the end we have to determine the minimum over all choices of $\mathcal C$, $\mathcal E$, $\pi$.
	
	In order to prove correctness,
	consider now a hypothetical geometric realization and an optimum solution of the ILP. Observe that the ILP is always feasible.
	It remains to show that the 2-class vertices can be inserted into
	the boomerangs and the 1-class vertices can be added so as to
	fulfill the alignment requirements prescribed by the ILP.  For each
	triplet $(v,b,d)$, where $b$ and~$d$ are two boomerangs incident to a
	common vertex cover vertex $v$, we add $x_{v,b,d}$ stubs of aligned
	edges in a close vicinity of $v$ inside the region reserved for $b$
	and~$d$. 
	Further, for each boomerang $d$ incident to a vertex cover vertex $v$ do the following.
	For each edge $e$ of $G_2$ that we required to be aligned with an edge in $d$ at $v$, we add a stub incident to $v$ into $d$ aligned with~$e$. We also add $y_{v,d}$ stups incident to $v$ into $d$.
	If the total number of stubs in all boomerangs
	representing contiguous subclasses of a 2-class $C$ are less
	than~$|C|$, we add more stubs into arbitrary boomerangs associated with~$C$.  
	We extend the stubs until the respective rays meet.  
	Observe that this intersection point is inside the boomerang
        and that no crossings are generated; see
        \cref{SUBFIG:insertingBoomerangs}.
\end{proof}

\section{Parameters: Segment and Line Cover Number}
\label{sec:segment-number}

In this section, we first study the parameterized complexity of computing the segment number of a planar graph with respect to the segment number and the line cover number separately. Recall that, if
$\seg(G)\le k$, then $\lin(G)\le k$. Then, we study the parameterized complexity of colored versions of the segment number and the line cover number with respect to their natural parameters.

We start by reviewing an FPT algorithm for the line cover number~\cite{cflrvw-cdgfl-JGAA23}.
A \emph{path component} (resp., \emph{cycle component}) of a graph is
a connected component isomorphic to a path (resp., a cycle).  The
algorithm first removes all the path components of the input graph
$G$, as they can be placed on a common line. 
Then each of the remaining connected components of $G$ is reduced by replacing long 
maximal paths. More precisely, maxmal paths of length greater than 
$\binom{k}{2}$ that contain only vertices of degree at most~$2$ are 
replaced by paths of length $\binom{k}{2}$, 
and cycle components with more than $\binom{k}{2}$ vertices are replaced by cycles of length $\binom{k}{2}$. 
Indeed, vertices of 
degree~2 are irrelevant since they can always be reintroduced by 
subdividing a straight-line segment of the path in a feasible solution 
of the reduced instance to obtain a feasible solution of the original instance. 
Note that the number of vertices of degree greater than $2$
must be at most $\binom{k}{2}$.  This yields an equivalent
instance~$G'$ with at most $1.1k^4$ vertices and $2k^4$ edges.
Next, the algorithm implicitly enumerates all line arrangements of $k$
lines.  For each such arrangement, the algorithm enumerates all
possible placements of the vertices of $G'$ at crossing points in the
arrangement. Observe that the vertices of degree greater than $2$ must
be placed at the crossing points of the lines. The instance is
accepted if at least one of the line arrangements can host $G'$.

In order to {\em explicitly} enumerate all line arrangements of $k$
lines (as needed in~\cref{algo:listIncidenceAlgo}), we can proceed as
described below. Observe that a line arrangement $\cal A$ of $k$ lines
defines a connected plane graph~$G_{\cal A}$ as follows.
The graph~$G_{\cal A}$ contains, for each crossing point~$p$
in $\cal A$, a vertex $u_p$, and, for every pair of crossing points
$p$ and $q$ that appear consecutively along a line of $\cal A$, an
edge~$\{u_p,u_q\}$.  Additionally, for each
half-line of $\cal A$ that starts at a crossing point~$p$, the graph
$G_{\cal A}$ contains a leaf adjacent to~$u_p$.  The clockwise order
of the edges around each vertex $u_p$ of
$G_{\cal A}$ (of degree larger than $1$) is naturally inherited from
the order of the crossing points of $\cal A$ around $p$. By
construction, $\cal A$ defines a {\em unique} connected plane
graph~$G_{\cal A}$ together with a {\em unique} covering of the
edges of $G_{\cal A}$ with $k$ edge-disjoint paths starting and ending
at leaves.  Thus, to enumerate all line arrangements of
$k$ lines, we execute the following steps. We consider all possible
planar graphs with at most $\binom{k}{2} + 2k$ vertices that contain $2k$
leaves, all their planar embeddings, and all their edge
coverings with $k$ edge-disjoint paths, if any. For each such
triplet (of a planar graph, an embedding, and an edge covering), we
test (using again Renegar's algorithm \cite{Renegar92a,Renegar92b})
whether the considered graph admits a planar straight-line drawing
preserving the selected embedding, in which each edge-disjoint path of
the considered covering is drawn on a straight-line. Each triplet that
passes the previous test corresponds to a (combinatorially different)
line arrangement of $k$ lines.

To compute the segment number of $G$, first observe that each path
component requires a segment in a drawing of $G$.  Assume that $G$ has
$p$ path components.  Then, let $k' = k-p$ be our new parameter and
remove the path components from~$G$.
Also, observe that the choices above only leave undecided to which
line of the arrangement the \emph{hanging paths} of $G'$, i.e., the
paths that start at a high degree vertex and end at a degree-$1$
vertex, are assigned. Clearly, the number of such choices are also
bounded by a function of $k$. Therefore, for every line arrangement
and every placement of the vertices of $G'$ to the crossing points of
the arrangement, we consider all possible ways to assign the
hanging paths to the lines and compute the actual number of segments
determined by these choices. We conclude that $\seg(G)\le k$ if and
only if we encountered a line arrangement and a placement of vertices
of $G'$ in this line arrangement that determines at most $k'$
segments. This immediately implies the following theorem.

\begin{theorem}\label{thm:segmentNumber}
	\textsc{Segment Number} is  FPT with respect to the natural parameter. 
\end{theorem}  

Note that the line cover number is lower or equal to the segment number. However, assuming no path component exists, we can show that the segment number is polynomially bounded by the line cover number.
\parametrizedProblem{\sc Segment Number by Line Cover Number}%
{A simple planar graph $G$, an integer~$s$.}%
{The line cover number $k=\lin(G)$ of $G$.}%
{Is the segment number of $G$ at most~$s$, that is, $\seg(G) \le s$?}

We exploit the following observation.

\begin{proposition}\label{obs:lineseg}
	Let $G_1,\ldots,G_t$ be the connected components of a graph
	$G$. The following hold:
	\begin{enumerate}
		\item $\seg(G)=\sum_{i=1}^t \seg(G_i)$.
		\item Let $G_1,\ldots,G_p$ be the path components of $G$ and let $G'=G\setminus\bigcup_{i=1}^p G_i$. Then, $\seg(G)=\seg(G')+p$. Moreover, if $V(G') \neq \emptyset$ then $\lin(G)=\lin(G')$, and $\lin(G)=1$ otherwise.
		\item If $G$ contains no path component, then $\seg(G)\le \lin^2(G)$.
	\end{enumerate}
\end{proposition}

\begin{proof}
	Items 1 and~2 are obvious.  For the last item, consider a
	drawing~$\Gamma$ of~$G$ on $k$ lines.  Since $G$ has no path
	component, every segment of~$\Gamma$ must contain at least one of the at most $(k-1)$ crossing points
	of the line that supports it. 
	Different segments supported by the same line are disjoint.
	Therefore, the number of segments of~$\Gamma$ cannot exceed $k(k-1)$.
\end{proof}

\begin{lemma}\label{the:segByLineCover}
	Computing the segment number is FPT parameterized by the line cover number of the input graph.
\end{lemma}

\begin{proof}
	Suppose that we are given an input graph $G$ with $\lin(G)=k$.  We
	first identify its connected components. Let $p$ be the number of path components of $G$.  Delete the path components of $G$, if any, to obtain a graph $G'$ for
	which $\seg(G')\le \lin^2(G') = \lin^2(G) = k^2$.  So the segment
	number of $G'$ is bounded by a function of $k$ and hence, by \cref{thm:segmentNumber}, can be computed in FPT time parameterized by $k$.  Once $\seg(G')$ is
	computed, output $\seg(G)=\seg(G')+p$, by Item~3 of \cref{obs:lineseg}.
\end{proof}

Motivated by list coloring, we generalize both  \textsc{Segment Number} and \textsc{Line Cover Number}  by prescribing admissible segments or lines for certain edges.
\parametrizedProblem{\sc List-Incidence Segment Number (Line Cover Number)}
{A simple planar graph $G$ and, for each $e \in E(G)$, a list $L(e)\subseteq [k]$.}%
{An integer $k$.}%
{Does there exist a planar straight-line drawing of $G$ and a labeling
	$\ell_1, \ell_2, \dots, \ell_k$ of the segments (supporting lines) of this
	drawing such that for every $e\in E(G)$, $e$ is drawn on a segment (line)
	in $\{\ell_i \colon i\in L(e)\}$?}

\begin{theorem}\label{thm:listlines-and-segments}
  The problems \textsc{List-Incidence Line Cover Number} and
  \textsc{List-Incidence Segment Number} are FPT with respect to the
  natural parameter.
\end{theorem}

Note that \cref{thm:listlines-and-segments} generalizes
\cref{thm:segmentNumber}, but the algorithm for \textsc{Segment
  Number}, which yields \cref{thm:segmentNumber}, is faster than the
algorithm for \textsc{List-Incidence Segment Number}.
In order to prove \cref{thm:listlines-and-segments}, we first consider the case where the input does not
contain any path or cycle components.  

\begin{figure}[tb]
	\centering
	\subcaptionbox{\label{fig:listIncA}}[0.495\linewidth]{\centering \includegraphics[page=2]{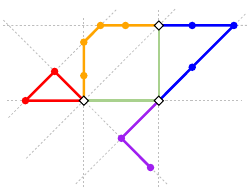}} 
	\hfil
	\subcaptionbox{\label{fig:listIncB}}[0.495\linewidth]{\centering \includegraphics[page=3]{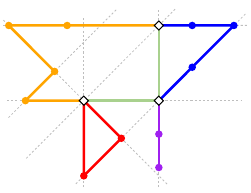}}
	\caption{\label{fig:listInc}
	Two routings of the light paths of $G$ on the same pair of a line arrangement and a vertex-placement of $G_{>2}$. The line arrangement is shown by \mbox{dotted} lines, the vertices of $G_{>2}$ are drawn as diamonds, and other vertices~are~drawn~as~\mbox{circles.}}
\end{figure}

\subparagraph{Graphs without Path and Cycle Components}

We first consider the {\sc List-Incidence Line Cover Number} problem.
Note that a feasible instance can have at most $\binom{k}{2}$ vertices of degree greater than $2$.
    Let $G_{>2}$ be the subgraph of $G$ induced by vertices of degree greater than $2$. We enumerate all possible
    arrangements of $k$ lines and all possible placements of vertices
    of $G$ of degrees greater than $2$ in the crossing points of the
    lines.
    The placement of vertices determines a
    straight-line drawing of $G_{>2}$, 
    (see green edges in \cref{fig:listInc}). We continue with the considered arrangement/vertex-placement pair only if the drawing of $G_{>2}$ is \emph{proper}, i.e., it is planar, and each edge of $G_{>2}$ belongs to a line of the arrangement and does not pass through another vertex of degree greater than $2$.
    A line segment that connects two consecutive crossing points 
    and does not correspond to an edge of $G_{>2}$ is called an
    \emph{gap}.
    Let $a$ and $b$ be two gaps. We say that $a$ is \emph{aligned} to $b$ (in
    this order) if they belong to the same line and
    the ending point of $a$ is the starting point of $b$. 
    Further, $a$ is \emph{adjacent} to $b$ (in this order) if
    they do not belong to the same line and the ending point of $a$ is
    the starting point of~$b$.
		
    Now consider a path $P=\langle u_0,u_1,\ldots,u_t \rangle$ in $G$ such that
    $\deg(u_0) > 2$, i.e., $u_0 \in V(G_{>2})$, $\deg(u_1) = \deg(u_2)
    = \dots = \deg(u_{t-1})=2$, and $\deg(u_t)$ is either $1$ or
    greater than $2$.  We call such a path \emph{light}.  Consider the
    representation of a light path in a solution.  The path leaves the
    vertex $u_0$ along
    some line, then turns to another line, possibly several times,
    until it either ends in a crossing point, where vertex $u_t$ had
    been placed (when $\deg(u_t)>2$) or in an inner point of some
    gap (when $\deg(u_t)=1$, note that w.l.o.g.\ we may assume
    that vertices of degree $1$ are not placed 
    in crossing points). This describes the \emph{routing} of $P$.
    Every maximal sequence of consecutive gaps that belong to the
    same line will be called a \emph{supergap} in the
    routing. In particular, we can represent a routing $R$ of $P$ by
    the sequence $\langle S_1, S_2, \dots, S_h \rangle$ of
    supergaps such that $S_i$ and $S_{i+1}$ are
    consecutive.  Every two consecutive gaps within a supergap 
    are aligned, while for two consecutive supergaps, the last
    gap of the first one is adjacent to the first gap of the
    second one (see \cref{fig:listInc}). 
    
    We use the following strategy. For every light path, we select its
    routing as a sequence of supergaps.  Furthermore, we enumerate
    over all possible $k!$ namings of the lines of the arrangement with
    labels in $\ell_1, \ell_2, \ldots, \ell_k$ and proceed only with
    those namings that determine a \emph{compatible drawing} of
    $G_{>2}$, i.e., one in which, for each edge $e$ of $G_{>2}$, it
    holds that $e$ lies on a line  $\ell_i$ with $i\in L(e)$.
    
    Let $\cal R$ be a collection of routings for the light paths. We check if $\cal R$ is \emph{feasible}, i.e., any two routings in $\cal R$ are non-crossing and internally disjoint, and each of them is internally disjoint with the drawing of $G_{>2}$. 
	For every light path, we check for each of its edges $e$ whether the list $L(e)$ allows its allocation along
    this routing. The first supergap $S_1$ of a routing must
    start in vertex $u_0$. The last one must end in $u_t$ if
    $\deg(u_t)>2$.  (If $\deg(u_t)=1$, there is no restriction on~$S_h$.)
    
    The strategy described above is summarized in
    \cref{algo:listIncidenceAlgo}.  To check whether the
    lists assigned to the edges of a path allow their allocation along a
    given routing, \cref{algo:listIncidenceAlgo} exploits the
    procedure \textsc{CheckRouting} (see
    \cref{algo:checkRoutingAlgo}), which works as follows.
    \begin{algorithm}[tb]
    	\SetKwInOut{Input}{Input}
    	\SetKwInOut{Output}{Output}
    	\Input{A planar graph $G$, an integer $k$, and for each $e \in E(G)$, a list $L(e)\subseteq [k]$.}
    	\Output{\textsf{true} if and only if there exists a planar straight-line drawing of $G$ and a labelling $\ell_1, \ell_2, \dots, \ell_k$ of the supporting lines of this drawing such that for every $e\in E(G)$, $e$ is drawn on a line in $\{\ell_i \colon i\in L(e)\}$.}
    	\ForEach{arrangement $\cal A$ of $k$ lines}
    	{
    		\ForEach{placement $\cal P$ of vertices of degree greater than 2 into crossing points of $\cal A$}
    			{
    				\If{$\cal P$ determines a proper drawing of $G_{>2}$}
    						{
    							\ForEach{collection of routings $\cal R$ of light paths of $G$}
    								{
    									\ForEach{naming of the lines in $\cal A$ with labels in $\ell_1, \ell_2, \ldots, \ell_k$}
    									{
	    									\If{$\cal R$ is feasible and the drawing of $G_{>2}$ is compatible
    										}
	    										{
	    											$X \leftarrow \mathsf{true}$;\\ 
	    											\ForEach{light path $P$ and its routing $R_P \in \cal R$}
	    												{
	    													$X \leftarrow X\,\land$
	                                                                                                        \textsc{CheckRouting}($P$, $R_P$, $L$)
	    												}
	    											\If{$X = \mathsf{true}$}
	    												{
	    													\Return{\textsf{true}}
	    												}
	    										}
    								}
    							}
    						}
    			}
    	}
	\Return{\textsf{false}}
        \caption{\textsc{ListLineCoverNumber}($G$, $k$, $L$)}
        \label{algo:listIncidenceAlgo}
\end{algorithm}
    For a light path $P=\langle u_0,u_1,\ldots,u_t \rangle$ and a routing
    $R=\langle S_1, S_2,\ldots, S_h \rangle$, we fill a table
    $CR(i,j)\in\{\mathrm{\sf true,false}\},i=1,\ldots,t, j=1,\dots,h$
    via dynamic programming. Its meaning is that $CR(i,j)=$ {\sf true}
    if and only if the subpath $P_i=P[u_0,\ldots,u_i]$ is \emph{realizable}
    along  $R_j = \langle S_1,\dots,S_j \rangle$, i.e., it admits a
    planar geometric realization such that all the vertices $u_1, u_2,
    \dots, u_i$ are placed inside the supergaps of $R_j$ and every 
    edge $e\in P_i$ lies on a line whose index is in the list~$L(e)$.
    
    \begin{algorithm}[tb]
    	\SetKwInOut{Input}{Input}
    	\SetKwInOut{Output}{Output}
    	\Input{A light path $P$, a routing $R_P$, and for each $e \in E(G)$, a list $L(e)\subseteq [k]$.}
    	\Output{\textsf{true} if and only if $P$ is realizable along $R_P$.}
    	// {\em Initialization of the table}\\
    	\For{$i=1$ \KwTo $t$}
    		{
    			\For{$i=1$ \KwTo $h$}
    				{
    					$CR(i,j) \leftarrow$ {\sf false}; 
    				}
    		}
    	$\ell_x \leftarrow$ the supporting line of $S_1$;\\
    	\If{$x \in L(\{u_0,u_1\})$}
    		{
    			$CR(1,1) \leftarrow$ {\sf true};
    		}
    	// {\em Update of the table}\\
    	\For{$i=1$ \KwTo $t-1$}
    	{
    		\For{$i=1$ \KwTo $h$}
    		{
    			$\ell_x \leftarrow$ the supporting line of $S_j$;\\
    			\If{$CR(i,j)$ \KwAnd $x \in L(\{u_i,u_{i+1}\})$}
    			{
    				$CR(i+1,j) \leftarrow$ {\sf true};
    			}
    			\If{$j < h$}
    			{
    				$\ell_y \leftarrow$ the supporting line of $S_{j+1}$;\\
    				\If{$CR(i,j)$ \KwAnd $y \in L(\{u_i,u_{i+1}\})$}
    				{
    					$CR(i+1,j+1) \leftarrow$ {\sf true};
    				}
    			} 
    		}
    	}
    	// {\em Return feasibility}\\
    	\Return{$CR(t,h)$};
    	\caption{\textsc{CheckRouting}($P$, $R_P$, $L$)}
    	\label{algo:checkRoutingAlgo}
    \end{algorithm}
    
    The correctness of \cref{algo:listIncidenceAlgo} follows from the
    fact that we are performing an exhaustive exploration of the
    solution space.  The correctness of \cref{algo:checkRoutingAlgo}
    follows from the fact that $P_i=\langle{}u_0,\dots,u_i\rangle$ can
    be realized along the initial part $\langle{}S_1,\dots,S_j\rangle$
    of $R$ so that vertex $u_i$ is placed inside the
    supergap~$S_j$ if and only if $P_i$
    can be realized so that $u_i$ is placed on the crossing point that
    is the last point of $S_j$.  The former is used (and needed) when
    updating $CR(i,j)$ to $CR(i+1,j)$, the latter for updating
    $CR(i,j)$ to $CR(i+1,j+1)$.
		
    The following lemma states that
    \cref{algo:listIncidenceAlgo} is FPT with respect to~$k$.
		
\begin{lemma}\label{cla:algoListLineNumber}
  Let $F(k)$ be the time needed to enumerate all non-isomorphic
  line arrangements of $k$ lines.  Then
  \cref{algo:listIncidenceAlgo} runs in time
  $O(F(k) k^{k^6+3k}) n$.
\end{lemma}

\begin{proof}
	Recall that there are $5$ nested \emph{for loops} in
	\cref{algo:listIncidenceAlgo}. For the running time estimate of
	\cref{algo:listIncidenceAlgo}, we note first that
	\cref{algo:checkRoutingAlgo} takes $O(t \cdot h\log k)$ time since
	the table has $t \cdot h$ entries and the update step requires
	searching through an at-most-$k$-element list. Since $t\le n$ and
	the number $h$ of supergaps on a routing cannot exceed the
	total number of gaps, which is no more than $k^2$, the running
	time of \cref{algo:checkRoutingAlgo} is upper bounded by
	$O(nk^2\log k)$. Since both the light paths and the routings in a
	collection of paths/routings are pairwise disjoint, $O(nk^2\log k)$
	is actually an upper bound on the total execution time of the
	innermost (fifth) {\em for loop} of \cref{algo:listIncidenceAlgo},
	where the procedure \textsc{CheckRouting} is called on every light
	path and a corresponding given routing for the path.

	Clearly, the number $Arr(k)$ of non-isomorphic line arrangements of $k$ lines is upper bounded by the number of planar graphs with at most $\binom{k}{2} + 2k$ vertices that contain $2k$ leaves i.e., $Arr(k) \in 2^{O(k^2)}$. Therefore, the enumeration of all such arrangements in the outermost (first) {\em for loop} can be performed in $F(k)$ time, where $F$ is a computable function~\cite{cflrvw-cdgfl-JGAA23}.

	Since there are at most $k\choose{2}$ crossing points of the lines, the vertices of degree greater than $2$ can be placed in at most ${k\choose{2}}!$ ways, which bounds the number of times the second {\em for loop} is executed.
	
	Inside the second {\em for loop}, we first check if $G_{>2}$ is proper. For this, we test for each edge of $G_{>2}$, if it belongs to a line of the arrangement and if it does not pass through any vertex of $G_{>2}$. There are $O(k^2)$ edges to be checked and on each of them at most $k-3$ other crossing points that could create a conflict, and so this check can be performed in $O(k^3)$ time. (This also includes checking if no two edges of $G_{>2}$ cross.)
	
	In a feasible instance, there is no more than $k^3$ light
	paths, because they start in one of at most $\binom{k}{2}$
	vertices and in each starting vertex, the first edge of the
	path has at most $2k$ choices (there are at most $k$ lines
	passing through this point, and on each of them two possible
	directions to be used). A routing of a light path is a
	sequence of at most $k(k-1)/2$ supergaps since each line
	may host at most $(k-1)/2$ disjoint supergaps.
	The first supergap on a routing can be chosen in at most
        $2k(k-1) \in O(k^2)$ ways (the routing chooses a line to
        follow, the direction, and the last point of this
        supergap). Similarly, the second supergap can be chosen in
        $O(k^2)$ ways, etc. Thus routing one path can be chosen in at
        most $(2k(k-1))^{k(k-1)/2} < (2k)^{k(k-1)} = k^{k(k-1)\log k}
        \in O(k^{k^2\log k})$ ways, and the collection of routing in
        $\big(O(k^{k^2\log k})\big)^{k^3} = O(k^{k^5\log k})$ ways.
	Thus the number of choices in the third {\em for loop} is at
        most $O(k^{k^5\log k})$.
	
	Since there exists at most $k!$ ways of the naming of the lines of the arrangement, the fourth {\em for loop} introduces a multiplicative factor of $k!$ in the running time.
	
	Therefore, the overall running time of \cref{algo:listIncidenceAlgo} is
	\[O\left(F(k){k\choose{2}}!\left(k^3 + k^{k^5\log k}k!nk^2\log
              k \right) \right) \subseteq O\left(F(k) k^{k^6+3k}\right) n.\]	
	Thus the algorithm is FPT when parameterized by~$k$.  
\end{proof}

	We now consider the {\sc List-Incidence Segment Number} problem. As already mentioned, any realization with $k$ segments is itself a realization with at most $k$ lines. Thus, we can use \cref{algo:listIncidenceAlgo}, with just two slight modifications. Observe that, in the third {\em for loop} of the algorithm, the considered collection of routings of the light paths together with the drawing of $G_{>2}$ uniquely determines the segments of the representation. Therefore, we only proceed to the fourth {\em for loop} with those collections that determine at most $k$ segments. However, if we proceed, in the fourth {\em for loop}, we consider all possible namings of these segments instead of considering all possible namings of the lines of the current arrangement. Clearly, the notion of compatible drawing of $G_{>2}$ is naturally modified using the segment labels.
	Finally, note that counting the number of segments affects the
        running time by a factor that depends solely on~$k$.
	
	\subparagraph{Handling Path and Cycle Components}
	
	To complete the proof of \cref{thm:listlines-and-segments}, we now show how to modify \cref{algo:listIncidenceAlgo} to handle cycle and path components.  Somewhat
	surprisingly, paths are the most troublemaking components. 
	
	First, we consider the cycle components. A cycle
	component $C$ cannot be drawn on a line, and thus any of its drawings
	must bend in a crossing point. Then it can be described by a closed
	routing of supergaps, and we include such a routing in the
        third {\em for loop} of \cref{algo:listIncidenceAlgo}.
	The subroutine \textsc{CheckRouting} is modified as follows: We pick one of the bending points and try all possibilities which vertex of the component $C$ is placed in this crossing point and in which direction the cycle $C$ travels around the chosen cyclic routing. This adds a multiplicative factor $2|C|=O(n)$ to the overall running time of the algorithm. 
	
	Second, we consider the path components. Let $P$ be such a component. 
	 If $P$ has a \emph{universal index}, i.e., if there is an index $i$ which belongs to all $L(e), e\in P$, we can realize this path on line $\ell_i$ in some sector which does not belong to any routing nor to edges of $G_{>2}$ (in case of {\sc List-Incidence Line Cover Number}) not affecting the realizability of  the rest of the graph. Thus, such paths can be ignored for {\sc List-Incidence Line Cover Number}.
	
	However, in the {\sc List-Incidence Segment Number} problem, it may be favourable to use bends even if the path has a universal index in the lists of its vertices, as this may prevent other components of the graph to be routed along the segment corresponding to the universal index.
	Moreover, even if we decide to route the path along a segment corresponding to one of the universal indices of the path, the feasibility of the instance may depend on which index we use. However, since the segment number of a graph equals the sum of the segment numbers of its connected components, the number of components of a feasible instance is always bounded by~$k$.
	Thus, for every path component~$P$, we check whether the set
        of universal indices, i.e., $\bigcap_{e \in P}L(e)$, is
        non-empty.  If this is the case, we (i)~decide whether we
        realize~$P$
	as a single segment, and if so, (ii)~what is the index of this
	segment.  This leaves $O(k)$ possibilities to try, and for each of
	them, we delete the index of this segment from all lists of other
	components.  Then we proceed with the rest of the graph.  Altogether, we have refined the search tree by a factor of $(k+1)!$.
	
	For paths with a universal index that were chosen to be represented with bends, and for paths that must be realized with bends (since they do not have any universal index), we use the routing approach (both in case
	of {\sc List-Incidence Segment Number} and {\sc List-Incidence Line Cover Number}) as described in the case of cycle components. Therefore, processing these components affects the running time by a multiplicative factor of $O(n)$. This concludes the proof of \cref{thm:listlines-and-segments}.

\section{Open Problems}
\label{sec:open}

We have shown that segment number parameterized by segment number,
line cover number, and vertex cover number is fixed parameter
tractable.  Another interesting parameter would be the treewidth.  So
far, even for treewidth~2, efficient optimal algorithms are only known
for some subclasses
\cite{Adnan2008,samee_etal:GD08,Goessmann_etal:WG22}.

The \emph{cluster deletion number} of a graph is the minimum number of
vertices that have to be removed such that the remainder is a union of
disjoint cliques.
Clearly, the cluster deletion number is always upper-bounded by the
vertex cover number.
Is the segment number problem FPT with respect to the cluster deletion
number?

\bibliographystyle{plainurl}
\bibliography{abbrv,segment-number}

\end{document}